\documentclass[11pt,a4paperwide]{amsart}
\usepackage[left=3.5cm, top=3cm, bottom=3cm, right=3.5cm]{geometry}
\usepackage{amsmath,amsfonts,amssymb,amsthm,amscd}

\usepackage{amsthm}
\theoremstyle{plain}

\newtheorem*{theoremA*}{Theorem A}
\newtheorem*{theoremB*}{Theorem B}

\usepackage{graphicx}
\usepackage{epsfig}

\makeatletter
\newcommand\xleftrightarrow[2][]{%
  \ext@arrow 9999{\longleftrightarrowfill@}{#1}{#2}}
\newcommand\longleftrightarrowfill@{%
  \arrowfill@\leftarrow\relbar\rightarrow}
\makeatother

\newtheorem{thm}{Theorem}
\newtheorem{lem}[thm]{Lemma}
\newtheorem{prop}[thm]{Proposition}
\newtheorem{defn}[thm]{Definition}
\newtheorem{cor}[thm]{Corollary}

\newtheorem{example}[thm]{Example}

\newtheorem{assumption}[thm]{Assumption}
\newtheorem{conjecture}[thm]{Conjecture}

\begin{document}

\title[Darboux coordinates for Novikov algebras]{Darboux coordinates for Hamiltonian structures defined by Novikov algebras }
\date{\today}

\author{Ian A.B. Strachan}

\address{School of Mathematics and Statistics\\ University of Glasgow\\Glasgow G12 8QQ\\ U.K.}
\email{ian.strachan@glasgow.ac.uk}

 \subjclass{}

\begin{abstract}
The Gauss-Manin equations are solved for a class of flat-metrics defined by Novikov algebras, this generalizing a result of Balinskii and Novikov who solved this problem in the case of commutative Novikov algebras (where the algebraic conditions reduce to those of a Frobenius algebra). The problem stems from the theory of first-order Hamiltonian operators and their reduction to a constant, or Darboux, form. The monodromy group associated with the Novikov algebra gives rise to an orbit space, which is, for a wide range of Novikov algebras, a cyclic quotient singularity.

\end{abstract}

\maketitle

\tableofcontents

\section{Introduction}

The existence, and construction, of solutions to the Gauss-Manin equations
\[
{}^g\nabla dv=0
\]
is a fundamental problem in many areas of mathematics - singularity theory, the theory of integrable systems, topological field theory, to name just a few \cite{Varchenko}. In this equation the connection ${}^g\nabla$ is the Levi-Civita connection for some
metric $g$, and the compatibility of this over-determined systems requires the metric to be flat. Solution of the Gauss-Manin equations then give the flat coordinates: the coordinates in which the metric coefficients are
constants. Such coordinates are also called Darboux coordinates - there are analogous to the existence of coordinates in which a non-degenerate sympletic form takes a constant, anti-diagonal form.

\medskip

The problem that will be addressed in this paper concerns the explicit construction of solutions to the Gauss-Manin equations for a class of metrics where the corresponding (inverse) metric is linear in the coordinates.
This problem has its origins in the theory of Hamiltonian structures (see section \ref{origins} below) and these structures are a direct multi-component generalization of the second Hamiltonian structure of the KdV hierarchy.
Such metrics are defined in terms of a so-called Novikov algebra $\mathcal{A}$ - a complex vector space equipped with a multiplication $\circ$ satisfying the defining properties
\begin{eqnarray*}
a \circ(b \circ c) - b \circ (a \circ c) & = & (a \circ b) \circ c - (b \circ a) \circ c\,,\\
(a \circ b) \circ c & = & (a \circ c) \circ b\,.
\end{eqnarray*}
Such algebras were first defined by Gelfand and Dorfman \cite{GD}, were studied further by Balinskii and Novikov \cite{BN}, and named Novikov algebras by Osborn \cite{O}.

\medskip The problem of finding the flat coordinates in the restricted case where the multiplication is commutative (where the above conditions reduce associativity equations) was solved by Balinskii and Novikov \cite{BN}. Here we solve the general case.

\medskip

This result of Balinskii and Novikov may be obtain by using the theory of commuting differential operators. The Gauss-Manin equations may be written
in the form
\[
L^{(i)} {\boldsymbol\xi} + \Lambda^{(i)} {\boldsymbol\xi}={\bf 0}\,.
\]
for certain vector fields $L^{(i)}$ and matrices $\Lambda^{(i)}\,.$
However in the general case these operators/matrices do not commute. What characterizes the Balinskii-Novikov result is that they restrict to the case where the operators do commute - one has an associated Abelian Lie algebra. The
fundamental
result that underlines the results in this paper is that, in general, the Lie algebra generated by these operators is solvable \cite{BK}. This rests on the fundamental result of Zelmanov \cite{Z} on the structural theory of Novikov
algebra,
this
answering a question raised in the original Balinskii-Novikov paper.

\medskip

The results of this papers also shows the potential for the development of a rich theory of solvable, as opposed to commutative, vector fields, within the theory of integrable systems and other areas \cite{CFG,Grab,Kaw}.

\subsection{Hamiltonian structures}\label{origins}

The origins of this problems stems from the study of Hamiltonian operators and the theory of bi-Hamiltonian systems. Given two functionals (defined under suitable boundary conditions)
\[
F=\int f({\bf u}\,,{\bf u}_X\,,\ldots ) \, dX\,,\qquad G=\int g({\bf u}\,,{\bf u}_X\,,\ldots ) \, dX
\]
one may define a bracket
\[
\{F,G\} = \int \frac{\delta F}{\delta u^i} \mathcal{H}^{ij} \frac{\delta G}{\delta u^j} \, dX
\]
and for this to be a Poisson bracket (so skew and satisfying the Jacobi equation) places conditions on the Hamiltonian operator $\mathcal{H}^{ij}\,.$

\medskip

The conditions for the first-order local operators
\begin{equation}
\mathcal{H}^{ij}({\bf u}) = g^{ij} ({\bf u}) \frac{d~}{dX} + \Gamma^{ij}_k({\bf u}) u^k_X
\label{hamiltonian}
\end{equation}
were first derived by Dubrovin and Novikov \cite{DN}.

\begin{thm}
Under the condition $\det(g^{ij} )\neq 0$ the conditions for $\mathcal{H}^{ij}$ to define a Poisson bracket are:
\begin{itemize}
\item[(a)] $g^{ij}$ is a symmetric tensor, and interpreted as the inverse of a (pseudo)-Riemannian metric with
associated torsion-free metric connection ${}^g\nabla$;
\item[(b)] $\Gamma^{ij}_k = - g^{ir} \Gamma^j_{rk}$, where $\Gamma^{i}_{jk}$ are the Christoffel symbols of the connection ${}^g\nabla$;
\item[(c)] the curvature of ${}^g\nabla$ is zero.
\end{itemize}
\end{thm}

\noindent Thus, by a fundamental result in differential geometry, there exists a coordinate system in which the components of the metric are constants - the so-called flat coordinate system, and
these are given by solutions of the corresponding Gauss-Manin system ${}^g\nabla d{\bf v} = {\bf 0}\,.$
In this system of coordinates the Hamiltonian operators takes the constant form
\[
\mathcal{H}^{ij}({\bf v}) = \eta^{ij} \frac{d~}{dX}\,.
\]

\medskip

The conditions for an operator of the form
\[
\mathcal{H}^{ij}({\bf u}) = \left( \Gamma^{ij}_k + \Gamma^{ji}_k\right) u^r \frac{d~}{dX} + \Gamma^{ij}_k u^k_X
\]
to define a Poisson bracket were first derived by Gelfand and Dorfman \cite{GD}\,. It is important to note that with specific metric, the symbols $\Gamma^{ij}_k$ turn out to be {\sl constants} for this class of metrics
(but not the symbols $\Gamma^i_{jk}$ which do depend on the $u^i$).
While their approach predates the above Theorem, their results may be quickly obtained as a special case of this more general result.
The geometric conditions for zero-curvature may be interpreted as purely algebraic conditions on the algebra defined by the constants $\Gamma^{ij}_k\,,$
\[
e^i \circ e^j = \Gamma^{ij}_k e^k
\]
and these conditions are the defining properties of a Novikov algebra.

\medskip

The theory extends to higher-order operators - the paradigm being the Miura transformation
\[
u = \frac{1}{2} v^2 - v_X
\]
that maps the second Hamiltonian operator of the KdV hierarchy to constant, or Darboux, form.
\begin{eqnarray*}
\mathcal{H}^{ij} (u) & = & \frac{d^3~}{dX^3} + 2 u \frac{d~}{dX} + u_X\,,\\
&&\\
& \updownarrow & u = \frac{1}{2} v^2 - v_X\\
&&\\
\mathcal{H}^{ij} (v) & = & \frac{d~}{dX}
\end{eqnarray*}
We will return to the discussion of higher-order operators in the final section. The problem that will be addressed here is the construction of the flat, or Darboux, coordinates for Hamiltonian operators which are homogeneous and linear in the fields - this being a direct multi-component generation of the (dispersionless) part of the second Hamiltonian structure of the KdV equation. Thus we need to solve the Gauss-Manin equations for the metric
\[
g^{-1} = \sum_{i,j,k} \left( \Gamma^{ij}_k + \Gamma^{ji}_k\right)  u^k \frac{\partial~}{\partial u^i} \otimes \frac{\partial~}{\partial u^j}
\]
where the constants $\Gamma^{ij}_k$ define a Novikov algebra.

\subsection{Outline}

The rest of the paper is laid out as follows. In Section \ref{algprelim} the formal definition of the various algebraic structures are given and various foundational results presented. In particular, the key result on which the
results of this papers rest - the solvability of a certain Lie algebra $\mathfrak{g}(\mathcal{A})$ constructed from the Novikov algebra - is given. This key result follows from the work of \cite{BG} and the structural theory of Novikov algebras developed in \cite{Z}. In Section \ref{GMeqn} the
Gauss-Manin equations are written as a first-order matrix system and the result of Balinskii-Novikov (the case of Abelian structures) rederived. In Section \ref{MainSection} a version of Frobenius' theorem for commuting vector fields is
extended to the case of vector fields that form a solvable (rather than Abelian) Lie algebra. Such a partial-straightening of the vector fields introduces a subsidiary set of coordinates - labeled $w^i$ - in which in Gauss-Manin equations can be solved.

\medskip

In Section \ref{MainSection2} the full problem is solved. In particular, under a non-degeneracy condition and the assumption of the existence of a right-identity, one obtains explicit formulae:

\begin{theoremA*}
Let $\mathcal{A}$ be a Novikov algebra with a right-identity and satisfying the non-degeneracy conditions in Assumption \ref{assumption}.  The transformation ${\bf u}={\bf u}({\bf v})$ is found by eliminating the ${\bf w}$-variables from the equations
\begin{eqnarray*}
v_i ({\bf w}) & = & \left( \prod^{\longrightarrow} e^{+ \Lambda^{(r)} w^r}\right)_{i1}\,,\\
u^i ({\bf w}) & = & \left( \prod^{\longleftarrow} e^{-\Xi^{(r)} w^r}\right)_{i\bullet}\,.
\end{eqnarray*}
where $v_i=\eta_{ij}v^j\,.$
\end{theoremA*}
\noindent Here $\Xi^{(i)}$ and $\Lambda^{(i)}$ are two sets of matrices that define representations of the solvable Lie algebra $\mathfrak{g}(\mathcal{A})\,.$
By consequences of Lie's theorem, these transformation are triangular and the remaining problem is just the linear algebraic problem of the elimination of the ${\bf w}$-coordinates to give $u^i=u^i({\bf v})\,.$

\medskip

These functions have various transformation properties with respect to the monodromy group, and these are explored in Section \ref{monodromysection}. With various natural assumptions one can derived stronger results.
The Gauss-Manin equation have a natural monodromy group which encodes the branching of the solutions around the
discriminant locus, which here corresponds to those points where $\det(g^{ij})=0\,.$ Under various natural assumptions this is a cyclic group which acts on the ${\bf v}$-space giving a quotient cyclic singularity.
This can be summarized in the main theorem:

\begin{theoremB*}
Suppose that $\lambda_i\in\mathbb{Z}$ for all $i=1\,,\ldots\,,n\,.$ With the conditions on the Novikov algebra contained in Assumption \ref{assumptions2}, the functions $u^i({\bf v})$ are invariant under the monodromy group and
\[
u^i({\bf v}) \in \mathbb{C}^{\mathcal{W}(\mathcal{A})} [ v^1\,,\ldots\,,v^{n-1}, v^n\,,\left(v^n\right)^{-1}]
\]
where the monodromy group is the cyclic group $\mathcal{W}(\mathcal{A})\cong \mathbb{Z}_{\lambda_n+1} [ 1=\lambda_1\,,\ldots\,,\lambda_n]$  which act on the $v_i$ variables by
\[
v_i({\bf w}) \mapsto \varepsilon^{\lambda_i} v_i({\bf w})\,.
\]
where $\varepsilon=e^{\frac{2\pi\sqrt{-1}}{\mu_n}}$ and $\mu_n=\lambda_n+1\,.$
The constants $\lambda_i$ satisfy the monodromy constraint $\lambda_i+\lambda_j=\lambda_n+1$ if $\eta_{ij} \neq 0\,.$

\medskip

In terms of the ${\bf w}$-variables the monodromy group acts as a translation in a single variable
\begin{equation}
\begin{array}{cclc}
w^1 & \mapsto & w^1 + \frac{2 \pi \sqrt{-1}}{\mu_n}\,,&\\
w^i & \mapsto & w^i\,, &\qquad\qquad i=2\,,\ldots\,,n\,.
\end{array}
\end{equation}

\end{theoremB*}

\noindent The specific case where the monodromy group is the cycle group $\mathbb{Z}_{n+1} [1\,,2\,,\ldots\,, n]$ is studied in more detail and it is shown that the function $u^i({\bf v})$ are homogeneous polynomials of degree $(i+1)\,.$ The monodromy group acts on the ${\bf v}$-space to form an orbit space, and this is related to the well-studied notion of a cyclic quotient singularity. This is explored in Section \ref{quotientsection}.
In the penultimate section, Section \ref{cotangent}, it is shown how the Novikov multiplication may be seen as a multiplication on the cotangent bundle and hence is a geometric structure independent of the specific
coordinate systems used in the previous sections. This utilizes the biHamiltonian structure obtained from the metric $g^{-1}$ and the cocycle $\eta^{-1}\,.$
Finally, in Section \ref{conclusion},  the question of a full dispersive Muira transformation is commented
on - this relates the results in this paper to recent work on Poisson cohomology.

\section{Algebraic preliminaries}\label{algprelim}

In this section we draw together the various definitions and basic results that will be used throughout this paper. We begin with the definition of an algebra, first given by Gelfand and Dorfman \cite{GD}.

\begin{defn}

\begin{itemize}
\item[(a)] Novikov algebra is a vector space $\mathcal{A}$ equipped with a composition (called multiplication) $\circ:\mathcal{A}\times\mathcal{A} \rightarrow\mathcal{A}$ with the properties

\begin{eqnarray}
a \circ(b \circ c) - b \circ (a \circ c) & = & (a \circ b) \circ c - (b \circ a) \circ c\,,\label{N1}\\
(a \circ b) \circ c & = & (a \circ c) \circ b\label{N2}
\end{eqnarray}
for all $a\,,b\,,c\in\mathcal{A}\,.$

\item[(b)] A cocycle on $\mathcal A$ is a symmetric bilinear map $\langle\,,\rangle\,: \mathcal{A} \times \mathcal{A} \rightarrow \mathbb{C}$ with the property
\[
\langle a \circ b , c \rangle =   \langle a, c \circ b\rangle
\]
for all $a\,,b\,,c\in\mathcal{A}\,.$

\end{itemize}

\end{defn}
Algebras with the single property (\ref{N1}) have been studies since the work of Cayley in the 19$^{\rm th}$ Century, and are known by a variety of names, for example, as:
\begin{itemize}
\item{} left-symmetric algebras;
\item{} pre-Lie algebras;
\item{} Koszul-Vinberg algebras.
\end{itemize}

\noindent The relations in part (a) may be written more succinctly in terms of left and right multiplication operators $L_a\,,R_b$ (defined via $L_ab=a\circ b = R_ba$ for all $a\,,b\in\mathcal{A}$):
\begin{eqnarray*}
\left[ L_a,L_b \right] & = & L_{[a,b]}\,,\\
\left[ R_a,R_b \right] & = & 0
\end{eqnarray*}
where $[a,b]=a\circ b - b \circ a\,.$ Thus a Novikov algebra is a left-symmetric algebra whose right multiplications commute.

\medskip

The following proposition is straightforward, but does not appear to have appeared in the literature before.

\begin{prop} Let $\omega\in\mathcal{A}^\star\,.$ Then
\begin{equation}
\langle a,b\rangle=\omega\left( a \circ b + b \circ a \right)
\label{cocycle}
\end{equation}
defines a cocycle on $\mathcal{A}\,.$
\end{prop}

\noindent The proof is by direct computation and uses both the defining properties (\ref{N1},\ref{N2}) of a Novikov algebra and will be omitted.

\medskip

\noindent The following is also immediate:

\begin{lem}

\begin{itemize}

\item[(a)] Let $\mathcal{A}$ be a pre-Lie algebra. Then the bracket
\[
[a,b]=a\circ b - b \circ a
\]
defines a Lie algebra $\mathfrak{g}(\mathcal{A})\,.$ In particular, the Lie algebra is Abelian if and only if the Novikov algebra is commutative (and hence associative).

\item[(b)] Let $\mathcal{A}$ be a Novikov algebra with cocycle $\langle\,,\rangle$ defined by (\ref{cocycle}). Then the cocycle is also a cocycle for the Lie algebra $\mathfrak{g}(\mathcal{A})\,,$
\[
\langle a,[b,c]\rangle +\langle b,[c,a]\rangle+\langle c,[a,b]\rangle=0\,.
\]

\end{itemize}

\end{lem}
\noindent Thus Novikov algebras are algebras whose left-multiplications forms a Lie-algebra $\mathfrak{g}(\mathcal{A})$ and whose right multiplications commute.

\medskip

Introducing a basis $\{ e^i\,, i=1\,,\ldots\,,n\}$ for $\mathcal{A}$ one may define the structure constants for the algebra
\[
e^i \circ e^j = \Gamma^{ij}_k e^k\,, \qquad \qquad i\,,j = 1\,, \ldots\,, n\,,
\]
and with these the structure constants of the Lie algebra $\mathfrak{g}(\mathcal{A})$ are given by
\begin{eqnarray*}
[e^i,e^j] & = & c^{ij}_k e^k \,,\\
& = & \left(\Gamma^{ij}_k - \Gamma^{ji}_k\right) e^k\,.
\end{eqnarray*}
It is clear that the map $L:\mathfrak{g}(\mathcal{A}) \rightarrow gl\left( \mathfrak{g}(\mathcal{A}) \right)$ given by
\[
a \mapsto L_a
\]
defines a representation of the Lie algebra $\mathfrak{g}(\mathcal{A})\,.$ However, there is no suitable representation theory for the Novikov algebra itself owing to the non-associativity of the multiplication.

Such a basis is only defined up to linear transformations. Using this freedom we define a distinguished basis element $e^\bullet$ with which the form $\omega$ takes the following form:
\begin{defn}
\[
\omega(a) = {\rm coefficient~of~} e^{\bullet} {\rm~in~expansion~of~}(a)\,,\qquad a\in \mathcal{A}\,.
\]
\end{defn}
\noindent It will be useful to define
\[
\eta^{ij} = \langle e^i,e^j\rangle\,.
\]
\noindent Throughout of the rest of this paper we assume:
\begin{assumption}\label{assumption}
There exists a basis element $e^\bullet$ with the property that
\[
\det\left( \eta^{ij} \right) \neq 0 \,.
\]
\end{assumption}
\noindent Under this assumption one can define $\eta_{ij} = \left(\eta^{ij}\right)^{-1}\,.$

\medskip

There are a number of different representation for the Lie algebra $\mathfrak{g}(\mathcal{A})$ that may be constructed from the structure constants of the Novikov algebra, and the interplay between these will be crucial
in the construction of Darboux coordinates.

\begin{prop} Given a basis $\{ e^i\,, i=1\,,\ldots\,,n\}$ for the Novikov algebra $\mathcal{A}$ and structure constants $\Gamma^{ij}_k$, define
matrices $\Lambda^{(i)}\,, \Xi^{(i)} \in gl\left( \mathfrak{g}(\mathcal{A}) \right)$ by:
\begin{eqnarray*}
\left( \Lambda^{(i)} \right)_{rc} & = & \Gamma^{ic}_r\,,\\
\left( \Xi^{(i)} \right)_{rc} & = & - \left( \Gamma^{ir}_c + \Gamma^{ri}_c\right)\,.
\end{eqnarray*}
Then
\begin{eqnarray}
\left[ \Lambda^{(i)},\Lambda^{(j)} \right] & = & c^{ij}_k \Lambda^{(k)}\,,\label{LieAlg1}\\
\left[ \Xi^{(i)},\Xi^{(j)} \right] & = & c^{ij}_k \Xi^{(k)}\label{LieAlg2}\,.
\end{eqnarray}
The Lie algebra $\mathfrak{g}(\mathcal{A})$ is Abelian if and only if $2 {\Lambda^{(i)}}^T + \Xi^{(i)}=0$ for $i=1\,,\ldots\,,n\,.$

\end{prop}

\begin{proof} The proof of (\ref{LieAlg1}) is immediate from the definition: in fact it only uses the first condition (\ref{N1}) and is just the proof that a pre-Lie algebra defines a Lie algebra.
The proof of (\ref{LieAlg2}) is more subtle and uses both of the Novikov conditions. The proof itself is straightforward and will be omitted: it is just direct computation. Both conditions are, in fact,
necessary for the proof of (\ref{LieAlg2}): the matrices $\Xi^{(i)}$ defined from the (non-Novikov) pre-Lie algebra $e^1\circ e^1 = 2 e^1\,, e^1 \circ e^2 = e^2\,, e^2 \circ e^2 = e^1\,,e^2 \circ e^1=0$ do not form a
Lie-algebra.

\medskip

Note, the left and right multiplications are related to these matrices by $L_{e^i} = {\Lambda^{(i)}}^T$ and $L_{e^i}+R_{e^i} = - \Xi^{(i)}\,.$ The condition $2 {\Lambda^{(i)}}^T + \Xi^{(i)}=0$ is then just
the statement that left and right multiplications coincide.

\end{proof}

Thus we have two representations, which we denote $\pi_{{}_\Lambda}$ and $\pi_{{}_\Xi}$, of the abstract Lie algebra $\mathfrak{g}(\mathcal{A})\,,$ defined by $\pi_{{}_\Lambda}(e^i)=\Lambda^{(i)}$
and $\pi_{{}_\Xi}(e^i) = \Xi^{(i)}\,.$

\begin{prop} The representation $\pi_{{}_\Xi}$ is faithful.
\end{prop}

\begin{proof} By definition,
\begin{eqnarray*}
\eta^{ij} & = & \omega(e^i \circ e^j + e^j \circ e^i)\,,\\
& = & ( \Gamma^{ij}_k + \Gamma^{ji}_k ) \omega(e^k)\,,\\
& = & \Gamma^{ij}_\bullet + \Gamma^{ji}_\bullet \,,\\
& = & - \Xi^{(i)}_{j\bullet}\,.
\end{eqnarray*}
Let ${\bf v}\in \mathfrak{g}(\mathcal{A}).$ Suppose $\pi_{{}_\Xi}(\bf v)=0\,.$ Then, on writing ${\bf v} = v_i e^i\,,$ one obtains $v_i \Xi^{(i)}=0$ and hence $v_i \Xi^{(i)}_{j\bullet}=0\,.$ Hence $v_i \eta^{ij}=0\,.$
By assumption (\ref{assumption}) it follows that $v_i=0$ and hence $\ker \pi_{{}_\Xi}=0\,.$
\end{proof}

The representation $\pi_{{}_\Lambda}$ is not, in general, faithful - see Example (\ref{nonfaithful}). It will be useful, in what follows, to define Lie algebra homomorphism
$\rho: gl\left(\mathfrak{g}(\mathcal{A})\right) \rightarrow gl\left(\mathfrak{g}(\mathcal{A})\right)$ by
\[
\rho\left( \Xi^{(i)} \right) = \Lambda^{(i)}\,.
\]
The interplay between these two representations of the Lie algebra $\mathfrak{g}(\mathcal{A})$ will play a pivotal role in the later constructions.

\medskip

Finally, the results of this paper will rest on the following result, proved in \cite{BK}, which uses the fundamental result of Zelmanov \cite{Z}.

\begin{thm} The Lie algebra $\mathfrak{g}(\mathcal{A})$ is solvable.
\end{thm}

\begin{proof} (Sketch, following \cite{BK}). The algebra $\mathcal{A}$ is right nilpotent if $R_A^m=0$ for some $n\geq 1$ (where $R_A=\{R_a\,: a \in \mathcal{A}\}$). Since right multiplications commute,
if $I$ and $J$ are two right nilpotent ideals it follows that their sum $I+J$ is also right nilpotent. Since $\mathcal{A}$ is finite dimensional there is thus a largest right-nilpotent ideal of $\mathcal{A}$, which we
denote $N(\mathcal{A})\,.$ From this it follows that the Lie algebra $\mathfrak{h}$ of $ N(\mathcal{A})$ is solvable \cite{S}. From Zelmanov's structural theorem \cite{Z}, $\mathcal{A}/N(\mathcal{A})$ is a direct sum of fields and
hence the Lie algebra $\mathfrak{g}(\mathcal{A})/\mathfrak{h}$ is Abelian and hence $\mathfrak{g}(\mathcal{A})$ is solvable.
\end{proof}

\section{The Gauss-Manin equations}\label{GMeqn}

Given a flat metric $g$ with Levi-Civita connection ${}^g\nabla$ the flat coordinates (coordinates in which the component of $g$ are all constants) are found by solving the Gauss-Manin equations
\[
{}^g\nabla d{\bf{v}} = {\bf 0}\,.
\]
Expanding this is terms of coordinates gives the system
\[
\frac{\partial^2 {\bf v}}{\partial u^j \partial u^k} - \Gamma^a_{jk}({\bf u}) \frac{\partial {\bf v}}{\partial u^a} = 0 \,,
\]
where $\Gamma^a_{jk}( {\bf u})$ are the Christoffel symbols of the metric $g({\bf u})\,.$ This is an over-determined systems of equations, but it forms a holonomic system due to the flatness of metric $g\,;$ this
geometric condition is precisely the analytic conditions for this over-determined system of equations to have an $n$-parameter solution space.

\medskip

We wish to solve the Gauss-Manin equations for the (flat) metric given in terms of an arbitrary Novikov algebra
\[
g^{ij}( {\bf u}) = \left( \Gamma^{ij}_k + \Gamma^{ji}_k \right) u^k\,.
\]
The corresponding Christoffel symbols are rational functions and it is difficult to say anything much about them in general. However, a more tractable, but equivalent, form may be found by multiplying the equations
by $g^{-1}$ to obtain
\[
g^{ik}({\bf u}) \frac{\partial~}{\partial u^k} \left(\frac{\partial{\bf v}}{\partial u^j}\right) - g^{ik}({\bf u}) \Gamma_{jk}^a({\bf u}) \frac{\partial {\bf v}}{\partial u^a} = 0 \,.
\]
By definition $\Gamma^{ia}_j = - g^{ik}({\bf u}) \Gamma^{a}_{jk}({\bf u})$ and these are, for metrics defined by Novikov algebras, {\sl constants}\,. Thus one obtains the system
\[
L^{(i)}({\bf u}) \frac{ \partial {\bf v} }{\partial u^j} + \Gamma_{j}^{ia} \frac{ \partial {\bf v}}{\partial u^a} = 0
\]
where the vector fields $L^{(i)}\,, i=1\,, \ldots\,, n$ are defined by
\begin{eqnarray*}
L^{(i)}({\bf u}) & = & g^{ij}({\bf u}) \frac{\partial~}{\partial u^j}\,,\\
& = &  \left( \Gamma^{ij}_k + \Gamma^{ji}_k \right) u^k \frac{\partial~}{\partial u^j}\,.
\end{eqnarray*}

\noindent The vector fields $L^{(i)}$ constitute a vector-field representation of the solvable Lie algebra $\mathfrak{g}(\mathcal{A})\,;$ it is easy to show that
\[
\left[ L^{(i)}, L^{(j)} \right ] = c^{ij}_k L^{(k)}
\]
where the bracket is the Lie bracket of vector fields.

\medskip

Thus with the introduction of these vector fields the Gauss-Manin equations becomes the first-order system
\[
L^{(i)} {\boldsymbol\xi} + \Lambda^{(i)} {\boldsymbol\xi} = 0\,,
\]
where
\[
{\boldsymbol \xi}({\bf u}) = \left(
\begin{array}{c}
\frac{\partial {\bf v}}{\partial u^1} \\
\vdots \\
\frac{\partial {\bf v}}{\partial u^n} \\
\end{array}
\right)\,.
\]
From the $n$-solutions ${\boldsymbol\xi}$ the flat coordinates may be found by simple quadrature.

\medskip

Solutions of the Gauss-Manin equations will exhibit branching around the discriminant
\[
\Sigma=\{ {\bf u}\, | \,\Delta({\bf u})=0 \}
\]
where $\Delta({\bf u}) = \det\left(g^{ij}({\bf u})\right),$ and this is encapsulated in the associated monodromy group
\[
W(M) = \mu \left( \pi_1(M\backslash\Sigma) \right)
\]
(see, for example, \cite{D}). Explicitly, the continuation of a solution under a closed path $\gamma$ on $M\backslash\Sigma$ yields a transformation
\[
{\tilde v}^a ({\bf u}) = A^a_b(\gamma) v^b({\bf u}) + B^a(\gamma)
\]
with $A$ orthogonal with respect to the metric $g({\bf u})\,,$ and these generate a subgroup of $O(N)\,.$ Thus to every Novikov algebra $\mathcal{A}$ there is an associated monodromy group which we
denote $W(\mathcal{A})\,.$

\begin{example} Consider the (commutative) Novikov algebra given by the multiplication table

\[
\begin{array}{c|cc}
\circ & e^1 & e^2 \\
\hline
e^1 & e^1 & e^2 \\
e^2 & e^2 & 0
\end{array}
\]
With this the (inverse)-metric is
\[
g^{-1} = 2 u^1 \frac{\partial~}{\partial u^1} \otimes \frac{\partial~}{\partial u^1} + 4 u^2 \frac{\partial~}{\partial u^1} \otimes \frac{\partial~}{\partial u^2}
\]
so $\Sigma=\{u^2=0\}\,.$ Solving the Gauss-Manin equations gives
\[
v^1 = {u^1}{ (u^2)^{-\frac{1}{2}} } \,, \qquad v^2 = 2  (u^2)^{\frac{1}{2}}\,.
\]
So under a closed path around $\Sigma$ given by $u^1\mapsto u^1\,, u^2 \mapsto e^{2\pi\sqrt{-1}} u^2\,,$
\begin{eqnarray*}
v^1 & \mapsto & e^{-\pi\sqrt{-1}} \, v^1\,,\\
v^2 & \mapsto & e^{+\pi\sqrt{-1}} \, v^2\,
\end{eqnarray*}
so
\[
A=\left(
\begin{array}{cc}
\varepsilon^{-1} & 0 \\
0 & \varepsilon \\
\end{array}
\right)\,,
\]
where $\varepsilon=e^{\pi\sqrt{-1}}\,.$ This gives a representation of the monodromy group $\mathbb{Z}_2\,.$ Inverting this transformation gives
\begin{equation}
u^1 =  \displaystyle{\frac{1}{2} v^1 v^2} \,,\qquad
u^2  =  \displaystyle{\frac{1}{4} (v^2)^2}\,,
\label{transformationexample}
\end{equation}
so $u_i \in \mathbb{C}^{\mathbb{Z}_2} [v_1\,,v_2]\,.$ Thus the $u_i$ are invariant polynomials under the action of the monodromy group.

\end{example}

In this example the Lie algebra $\mathfrak{g}(\mathcal{A})$ is Abelian, and the transformation (\ref{transformationexample}) is quadratic. This is true more generally - if $\mathfrak{g}(\mathcal{A})$ is Abelian the
transformation between the ${\bf v}$-coordinates and the ${\bf u}$-coordinates will {\sl always} be quadratic. This was proved by Novikov and Balinskii by direct calculation. Basically, if $\mathfrak{g}(\mathcal{A})$ is Abelian one can
circumvent the solving of the Gauss-Manin equations and the inverting of its solution by constructing a suitable quadratic ansatz and directly verifying that this has the required properties. However, simple examples with
non-Abelian Novikov algebras show that such quadratic ansatz will not hold in general - indeed, one does not know a priori the degrees of the polynomials, or even if the functions are polynomial.

\medskip

Low-dimensional Novikov algebras and corresponding cocycles where classified by Bai and Meng \cite{BM1,BM1b} and provide a rich sources of examples.

\begin{example}\label{2Dexample} Consider the Novikov algebra given by the multiplication table

\[
\begin{array}{c|cc}
\circ & e^1 & e^2 \\
\hline
e^1 & e^1 & \lambda e^2 \\
e^2 & e^2 & 0
\end{array}
\]
\noindent Solving the Gauss-Manin equations and inverting gives

\[
u^1 = \frac{1}{2} v^1 v^2\,, \qquad \qquad u^2 = \frac{1}{4} \left( v^2 \right)^{1+\lambda}\,.
\]
\noindent These functions are only polynomial if $\lambda$ is an integer, in which case they are invariant under the cyclic group generated by transformation $v^1 \mapsto \varepsilon^\lambda v^1\,,v^2 \mapsto \varepsilon v^2$
where $\varepsilon^{\lambda+1}=1$ and hence $u^1\,,u^2 \in \mathbb{C}^{\mathbb{Z}_{\lambda+1}}[v^1,v^2]\,.$

\medskip

If $\lambda=p/q$ is rational, one obtains the monodromy group $\mathbb{Z}_{p+q}[p,q]$ acting on the ${\bf v}$-coordinates via $v^1 \mapsto \varepsilon^p v^1\,,v^2 \mapsto \varepsilon^q v^2$ where $\varepsilon^{p+q}=1\,.$
Note, though, that the functions ${\bf u}={\bf u}({\bf v})$ are not polynomial functions, so a finite monodromy group does not, by itself, imply that these are polynomial functions.

\end{example}

While the aim of this paper is the construction of solution of the Gauss-Manin equations in general, we first recover the Balinskii-Novikov result - this will illustrate the general method that will presented in
Section \ref{MainSection}

\subsection{The Balinskii-Novikov case: Abelian Lie algebras}

If the Novikov algebra $\mathcal{A}$ is commutative the associated Lie algebra $\mathfrak{g}(\mathcal{A})$ is Abelian and hence the vector field $L^{(i)}$ commute. Hence, by Frobenius' Theorem there exists coordinates
$\{w^i\,, i=1\,,\ldots\,, n\}$ with the property that
\begin{equation}
L^{(i)} = \frac{\partial~}{\partial w^i}\,, \qquad i=1\ldots\,, n\,.\label{changeofvariable}
\end{equation}
In these coordinates the Gauss-Manin equation just become the matrix partial differential equations
\[
\frac{\partial {\boldsymbol{\xi}}}{\partial w^i} + \Lambda^{(i)} {\boldsymbol \xi} = 0
\]
which may be trivially solved to yield
\[
{\boldsymbol \xi}({\bf w}) = \left( \prod_{r=1}^n e^{-\Lambda^{(r)} w^r} \right) {\boldsymbol \xi}_0
\]
where ${\boldsymbol \xi}_0$ is a constant vector. Since the Lie algebra is Abelian, the matrices $\Lambda^{(i)}$ all commute so there is no ambiguity or ordering problem in the matrix-exponentials. Thus it is trivial to
solve the Gauss-Manin equations in the wrong coordinate system!

It follows from (\ref{changeofvariable}) that the ${\bf u}$-coordinates and the ${\bf w}$-coordinate systems are related via the differential equation
\begin{eqnarray*}
\frac{\partial u^j}{\partial w^i} &=& 2 \Gamma^{ij}_s u^s\,,\\
&=& - \left[ \Xi^{(i)} \right]_{js} u^s
\end{eqnarray*}
or, as a matrix system,
\[
\frac{\partial {\bf u}}{\partial w^i} = - \Xi^{(i)} {\bf u}\,,
\]
where
\[
{\bf u} = \left(
\begin{array}{c}
u^1 \\
\vdots \\
u^n \\
\end{array}
\right)\,.
\]
Since the algebra $\mathfrak{g}({\mathcal{A}})$ is Abelian, the matrices $\Xi^{(i)}$ commute and hence one obtains
\[
{\bf u}({\bf w}) = \left( \prod_{r=1}^n e^{-\Xi^{(r)} w^r} \right) {\bf u}_0
\]
for some constant vector ${\bf u}_0\,.$ Recall that, since $\mathfrak{g}(\mathcal{A})$ is Abelian, the two sets of matrices are related via $2{\Lambda^{(i)}}^{T}+\Xi^{(i)}=0\,.$

\begin{thm}\label{BNtheorem} Suppose $\mathfrak{g}(\mathcal{A})$ is Abelian. Then:

\begin{itemize}

\item[(a)] In the ${\bf v}$-coordinates,
\[
g^{-1} = \eta^{ij} \frac{\partial~}{\partial v^i} \otimes \frac{\partial~}{\partial v^j}\,,
\]
where $\eta^{ij}$ is the distinguished cocycle defined above;

\item[(b)] The transformation from the $\bf v$ to the $\bf u$ coordinates is given by
\[
u^i = \frac{1}{2} \Gamma^i_{jk} v^j v^k
\]
where $\Gamma^i_{jk} = \eta_{jr} \Gamma^{ir}_k\,.$

\end{itemize}

\end{thm}

\begin{proof} The approach here will be to construct solution to the Gauss-Manin equation in the $\bf w$-coordinates, and use this, together with the transformation from the $\bf w$-coordinates to the $\bf u$-coordinates, to
prove the result. As we are dealing with systems of partial differential equations, boundary conditions need to be fixed: these will be defined at the point ${\bf w}={\bf 0}\,.$

\medskip

We first order and normalize the $\bf v$ coordinates (effectively a choice of the vector ${\boldsymbol \xi}_0$ for each component of $\bf v$), so that
\begin{equation}
\frac{\partial v^j}{\partial u^i} = \left(
\prod_{r=1}^n e^{-\Lambda^{(r)} w^r} \right)_{ij}\,.
\label{vdu}
\end{equation}
We fix boundary conditions so
\[
\left. u^i \right|_{{\bf w}={\bf 0}} = \delta_{i\bullet}\,,
\]
and this in terms fixes ${\bf u}_0$, so $u^i_0 = \delta_{i\bullet}\,.$
With this choice,
\[
\left( \Xi^{(i)} {\bf u}_0 \right)_j = \Xi^{(i)}_{j\bullet} = - g^{ij}_\bullet = -\eta^{ij}
\]
and it then follows that
\begin{equation}
\frac{\partial u^i}{\partial w^j} = \left(
\prod_{r=1}^n e^{-\Xi^{(r)} w^r} \right)_{ir} \eta^{rj}\,.
\label{udw}
\end{equation}
Note: both the matrices (\ref{vdu}) and (\ref{udw}) may be easily inverted.

\medskip

\noindent{\sl Proof of (a)} Since, by construction, the components of
\[
g^{-1}({\bf v}) = \left(g^{ab}({\bf u}) \frac{\partial v^i}{\partial u^a} \frac{\partial v^j}{\partial u^b} \right)\, \frac{\partial~}{\partial v^i} \otimes \frac{\partial~}{\partial v^j}
\]
are constant, it suffices to evaluate them at the specific point ${\bf w}={\bf 0}\,.$ Using the above formulae
\begin{eqnarray*}
g^{ij}({\bf v})\vert_{  {\bf w} = {\bf 0} } & = & g^{ab}_r \delta_{r\bullet} \delta_{ia} \delta_{jb}\,,\\
& = & -\Xi^{(i)}_{j\bullet}\,,\\
& = & \eta^{ij}\,.
\end{eqnarray*}
Hence the result. This proof also works in the general, non-commutative case, on using the analogous formulae in the Propositions \ref{udefprop} and \ref{GMsolution}.

\medskip

\noindent{\sl Proof of (b)} Using the chain rule
\[
\frac{\partial~}{\partial v^i} = \frac{\partial w^a}{\partial v^i} \, \frac{\partial u^j}{\partial w^a} \frac{\partial~}{\partial u^j}\,.
\]
one may show (and again this crucially uses the commutativity of the matrices) that
\[
\frac{\partial^2 u^i}{\partial v^a \partial v^b} = \Gamma^{i}_{ab}\,.
\]
Integrating, and imposing the boundary condition
\[
\left. v^i \right|_{{\bf w}={\bf 0}} = 2 \left. u^i \right|_{{\bf w}={\bf 0}} = \delta_{i\bullet}\,,
\]
which eliminates linear and constant terms, yields the result. The factor $2$ comes from the various normalizations used.

\end{proof}

\bigskip

Note, such a direct computational approach relies on the commutativity of the Novikov algebra. In general (and this will be expanded on in the next section) one just obtains explicit formulae, in terms of
the $\bf w$-coordinates, for the functions ${\bf u}({\bf w})$ and $\frac{\partial {\bf v}}{\partial {\bf u}}({\bf w})$ which cannot be inverted in general.

\medskip

A shorter, alternative proof of this Theorem was presented in \cite{SZ}. This used a lifting procedure to generate the transformation from the the $1$-dimensional transformation $u=\frac{1}{2} v^2\,.$

\medskip

The results in this subsection rely on the existence of the ${\bf w}$-coordinate system whose existence is implied, via the commutativity of the Novikov algebra, by Frobenius' Theorem for commuting vector fields.
Frobenius' Theorem is often heuristically interpreted in terms of \lq straightening out\rq~of vector fields. In the general case, such straightening is not possible as the vector fields $L^{(i)}$ do not commute.
However, from the solvability of the Lie algebra $\mathfrak{g}(\mathcal{A})$ a {\sl partial} straightening may be constructed and the coordinate system in which such a partial straightening takes place may be used
in the same way as the $\bf w$-coordinates were used in this section.

\section{A Frobenius Theorem for solvable vector fields}\label{MainSection}

Recall that the Lie algebra $\mathfrak{g}(\mathcal{A})$ is solvable. Thus by definition the derived, or commutator, series
\begin{eqnarray*}
\mathfrak{g}^{(0)} & = & \mathfrak{g}(\mathcal{A})\,, \\
\mathfrak{g}^{(i)} & = & \left[\mathfrak{g}^{(i-1)},\mathfrak{g}^{(i-1)}\right]\,,
\end{eqnarray*}
forms a decreasing sequence
\[
\mathfrak{g}(\mathcal{A})=\mathfrak{g}^{(0)} \supseteq \mathfrak{g}^{(1)}\supseteq\mathfrak{g}^{(2)} \supseteq \ldots
\]
which terminates: $\mathfrak{g}^{(m)}=0$ for some $m\,.$ The following decomposition of a solvable Lie algebra will be central:

\begin{prop}\label{propsequence} [See, for example, \cite{K}.] An $n$-dimensional Lie algebra $\mathfrak{g}$ is solvable if and only if there exists a sequence of subalgebras
\begin{equation}
\mathfrak{g} = \mathfrak{a}_0 \supseteq \mathfrak{a}_1 \supseteq \ldots \supseteq \mathfrak{a}_n=0
\label{sequence}
\end{equation}
such that, for each $i$, $\mathfrak{a}_{i+1}$ is an ideal in $\mathfrak{a}_i$ and $\dim\left( \mathfrak{a}_i  / \mathfrak{a}_{i+1} \right)=1\,.$
\end{prop}

In terms of the representation $\pi_{{}_\Xi}$ we may change basis so $\mathfrak{a}_{i-1} = \mathbb{C}\Xi^{(i)} \oplus \mathfrak{a}_{i}$, so, as a vector space,
\[
\mathfrak{a}_i = {\rm span}\{ \Xi^{(j)}\,, j=i+1\,,\ldots\,,n\}\,.
\]
Note, for notational convenience we drop the dependence on the representation and use the same symbol $\mathfrak{a}_i$ to denote both $\mathfrak{a}_i$ and $\pi_{{}_\Xi}(\mathfrak{a}_i)\,.$ Thus $\mathfrak{a}_{i-1}$ is a
semi-direct product of $\mathfrak{a}_{i}$ and the one-dimensional Lie algebra $\mathbb{C} \Xi^{(i)}\,.$

\medskip

The sequence (\ref{sequence}), by Lie's Theorem, results in the existence of an invariant flag of subspaces,
\[
V=V_0 \supseteq V_1  \ldots \supseteq V_n=0,.
\]
But this sequence also has a geometric interpretation. Since we have a representation in terms of vector fields, from $[\mathfrak{a}_i, \mathfrak{a}_i] \subset \mathfrak{a}_{i+1}$, the vector fields in $\mathfrak{a}_i$ form
an integrable distribution, and hence define a submanifold. Thus one obtains, analogous to the invariant flag of subspaces, an invariant flag of nested submanifolds.

\medskip

To illustrate the general construction we consider the vector fields from Example \ref{2Dexample},
\begin{eqnarray*}
L^{(1)} & = & 2 u^1 \frac{\partial~}{\partial u^1} + 3 u^2 \frac{\partial~}{\partial u^1}\,, \\
L^{(2)} & = & 3 u^2 \frac{\partial~}{\partial u^1}
\end{eqnarray*}
so $\left[ L^{(1)},L^{(2)} \right] = L^{(2)}\,.$ Any single vector field may be straightened out so we start with the vector field from the subalgebra $\mathfrak{a}_1$ at the {\sl end} of the elementary
sequence (\ref{sequence}). Thus coordinates may be found so that
\[
3 u^2 \frac{\partial~}{\partial u^1} = \frac{\partial~}{\partial w^2}\,.
\]
We now introduce a vector field ${\bf v} = L^{(1)} - \alpha(w^1,w^2) L^{(2)}$ and fix the scalar function $\alpha$ by requiring that $\left[ {\bf v}, L^{(2)}\right]=\left[ {\bf v}, \partial_{w^2} \right]=0\,.$ This
gives $\alpha=-w_2\,.$ Hence from these commuting vector fields one may introduce coordinates so
\begin{eqnarray*}
L^{(1)} & = & \frac{\partial~}{\partial w^1} - w^2 \frac{\partial~}{\partial w^2}\,, \\
L^{(2)} & = & \frac{\partial~}{\partial w^2}
\end{eqnarray*}
and a simple calculation gives $u^1 = w^2 e^{3 w^1}\,, u^2 = e^{3 w^1}\,.$

\medskip

Thus the structure of the elementary sequence gives an ordering which may be used to construct the $\bf w$-coordinate system which partially straightens out the vector fields $L^{(i)}\,.$ This is entirely analogous with
Lie's original integration method of differential equations: the decomposition of the Lie algebra determines the integration scheme.

\medskip

\begin{defn} Recall the decomposition $\mathfrak{a}_{i-1} = \mathbb{C}\Xi^{(i)} \oplus \mathfrak{a}_{i}$. The matrix-valued functions $g_{(i)}$ are defined by:

\begin{equation}
g_{(i)}(w^{i+1}\,,\ldots\,,w^{n}) =
\left\{
\begin{array}{ll}
\displaystyle{ \prod_{   \{  r:  \, \Xi^{(r)} \in \mathfrak{a}_{i}  \}  }^{\longleftarrow} e^{-\Xi^{(r)} w^r}\,,} & i=1\,,\ldots\,,n-1\,, \\
&\\
\displaystyle{\mathbb{I}} & i=n\,;\\
\end{array}
\right.
\label{gdef}
\end{equation}
\noindent The scalar functions $\alpha^{(i)}_r$ are defined as the coefficients in the expansion:
\begin{equation}
\Xi^{(i)} - g_{(i)} \Xi^{(i)} g_{(i)}^{-1} = \sum_{    \{  r:  \, \Xi^{(r)} \in \mathfrak{a}_{i}  \}  } \alpha^{(i)}_r ( w^{i+1}\,, \ldots\,,w^n) \,\Xi^{(r)}\,.
\label{defalpha}
\end{equation}
\end{defn}
\noindent Two remarks are in order. Firstly, the notation $\displaystyle{\prod^{\longleftarrow}}$ denotes the order of the terms in the direction of increasing labels so
\[
g_{(i)} = e^{-\Xi^{(n)} w^{n}} \cdot \ldots \cdot  e^{-\Xi^{(i+1)} w^{i+1}}\,.
\]
Since the $\Xi^{(i)}$ do not commute, specifying the ordering is essential. Secondly, from the extension $\mathfrak{a}_{i-1}=\mathbb{C} \Xi^{(i)} \oplus \mathfrak{a}_{i}$, we are conjugating $\Xi^{(i)}$ by an
element in $e^{\mathfrak{a}_{i}}\,.$ Hence the left-hand side of (\ref{defalpha}) must lie in $\mathfrak{a}_{i}$ (the leading order terms that lie in $\mathbb{C} \Xi^{(i)}$ cancel). Thus the sum on the right-hand side of (\ref{defalpha}) is
over terms in $\mathfrak{a}_{i}$, i.e. $\sum_{r>i}\,.$

\medskip
We begin by constructing the ${\bf w}$-coordinate system in which the vector fields $L^{(i)}$ are partially straightened out.

\begin{prop}\label{udefprop}
Let
\begin{equation}
{\bf u} = \left(
\prod_{r=1\,,\ldots\,,n}^{\longleftarrow} e^{-\Xi^{(r)} w^r}
\right) {\bf u}_0\,.
\label{udef}
\end{equation}
Then, in the ${\bf w}$-coordinates,
\[
L^{(i)} = \frac{\partial~}{\partial w^i} + \sum_{r>i} \alpha^{(i)}_r L^{(r)}\,.
\]
\end{prop}

\begin{proof} On differentiating (\ref{udef})\,,
\begin{eqnarray*}
\frac{\partial{\bf u}}{\partial w^i} & = & - \left(\prod_{r=i+1\,,\ldots\,,n}^{\longleftarrow} e^{-\Xi^{(r)} w^r} \right) \, \Xi^{(i)} \,
\left(\prod_{r=1\,,\ldots\,,i}^{\longleftarrow} e^{-\Xi^{(r)} w^r} \right) {\bf u}_0\,,\\
& = & - \left( g_{(i)} \Xi^{(i)} g_{(i)}^{-1} \right) {\bf u}
\end{eqnarray*}
and on using (\ref{defalpha}),
\begin{eqnarray*}
\frac{\partial u^j}{\partial w^i} & = & \left( \sum_{r>i} \alpha^{(i)}_r \Xi^{(r)}_{jk} - \Xi^{(i)}_{jk} \right) u^k \,,\\
& = & -\sum_{r>i} \alpha^{(i)}_r g^{rj} + g^{ij}\,.
\end{eqnarray*}
Hence
\begin{eqnarray*}
\frac{\partial~}{\partial w^i} + \sum_{r>i} \alpha^{(i)}_r L^{(r)} & = & \left(-\sum_{r>i} \alpha^{(i)}_r g^{rj} + g^{ij}\right) \frac{\partial~}{\partial u^j} + \sum_{r>i} \alpha^{(i)}_r L^{(r)}\,, \\
& = & L^{(i)}
\end{eqnarray*}
as required.
\end{proof}

\medskip

Since the transformation between the $L^{(i)}$ and the $\frac{\partial~}{\partial w^i}$ is triangular one may easily invert these equations to get expressions for the $L^{(i)}$ as linear combinations
of the $\frac{\partial~}{\partial w^j}\,,$ but such expressions are not, in general, required. However they may be calculated very easily, as the following examples shows.

\medskip

\begin{example}\label{2stepexample}

Suppose, given the sequence (\ref{sequence}), that $\mathfrak{a}_1$ is Abelian, so $\mathfrak{g}(\mathcal{A})$ is a 2-step solvable Lie algebra.
Since the matrices $\Xi^{(i)}\,,i\geq 2$ commute we may write
\[
g_{(1)} = e^{-\sum_{r\geq 2} \Xi^{(r)} w^r}\,.
\]
To calculate $\alpha^{(1)}_i$ (all others are zero), we write
\[
\Xi^{(1)} - g_{(1)} \Xi^{(1)} g_{(1)}^{-1} = \left[ \Xi^{(1)}, g_{(1)} \right] g_{(1)}^{-1}
\]
and use the formula (which holds since $\left[\Xi^{(1)}, -\right]$ is a derivation on $\mathfrak{a}_1$),
\[
\left[ \Xi^{(1)}, e^A \right] = \int_{s=0}^1 e^{sA} \left[ \Xi^{(1)},A \right] e^{(1-s)A}\,ds\,.
\]
With this
\begin{eqnarray*}
\left[\Xi^{(1)},g_{(1)}\right] & = & -\sum_{i>1} w^i \int_{s=0}^1 e^{-s\sum_a \Xi^{(a)} w^a} \left[ \Xi^{(1)},\Xi^{(i)} \right] e^{-(1-s)\sum_b \Xi^{(b)} w^b}\,,\\
& = & - \sum_{i,r>1} w^i \,c^{1i}_r \int_{s=0}^1 e^{-s\sum_a \Xi^{(a)} w^a} \Xi^{(r)} e^{-(1-s)\sum_b \Xi^{(b)} w^b}\,,\\
& = & - \sum_{i,r>1} w^i \,c^{1i}_r \,\Xi^{(r)}\, g_{(1)}\,.
\end{eqnarray*}
Hence
\[
\alpha^{(1)}_r (w_2\,,\ldots\,,w_{n}) = - \sum_{i>1} w^i c^{1i}_r\,.
\]
Hence, in the partially straightened out coordinates,
\begin{eqnarray*}
L^{(1)} & = & \frac{\partial~}{\partial w^1} - \sum_{i,r>1} w^i c^{1i}_r \frac{\partial~}{\partial w^r}\,,\\
L^{(i)} & = & \frac{\partial~}{\partial w^i}\,, \qquad i>1\,.
\end{eqnarray*}
The 2-dimensional Novikov algebras from Example (\ref{2Dexample}) falls into this class of examples.

\end{example}

\noindent In these coordinates one may solve the Gauss-Manin equations.

\section{Solutions of the Gauss-Manin equations for Novikov algebra}\label{MainSection2}

The only difference in the formula for ${\bf u}({\bf w})$ between the commutative and the general, non-commutative, case is that a precise order of the exponential factors is required. This is also the case for the formula
for ${\boldsymbol \xi}({\bf w})\,.$

\begin{prop}\label{GMsolution}
The function
\[
{\boldsymbol \xi}({\bf w}) = \left( \prod_{r=1\,,\ldots\,, n}^{\longleftarrow} e^{-\Lambda^{(r)} w^r} \right) {\boldsymbol \xi}_0\,,
\]
where ${\boldsymbol \xi}_0$ is a constant vector, satisfies the Gauss-Manin equations

\[
L^{(i)} {\boldsymbol \xi} + \Lambda^{(i)} {\boldsymbol \xi} = {\bf 0}\,, \qquad i=1\,, \ldots \,, n\,.
\]
\end{prop}

\medskip

\begin{proof}

We prove this result by extension, using the decomposition $\mathfrak{a}_{i-1} = \mathbb{C}\Lambda^{(i)} \oplus \mathfrak{a}_{i}$ (this decomposition follows from using the Lie algebra homomorphism $\rho\left( \Xi^{(i)} \right) = \Lambda^{(i)}$), working up the sequence (\ref{sequence}) from the end. Let
\[
h_{(i)} = \prod_{\{r\,: \Lambda^{(r)} \in \mathfrak{a}_i\}}^{\longleftarrow} e^{\Lambda^{(r)} w^r}\,,
\]
so $h_{(i-1)} = h_{(i)} e^{-\Lambda^{(i)} w^i}\,.$

\medskip

We fix $i$ and assume that
\[
L^{(j)} h_{(i)} = - \Lambda^{(j)} h_{(i)}
\]
for each $j>i\,.$This is clearly true at the end of the sequence when $i=n-1\,:$
\begin{eqnarray*}
L^{(n)} h_{(n-1)} & = & \frac{\partial~}{\partial w^{n}} \left(e^{-\Lambda^{(n)} w^{n}}\right)\,, \\
& = & - \Lambda^{(n)} h_{(n-1)}\,.
\end{eqnarray*}
Now consider $L^{(j)} h_{(i-1)}$ for $j > i-1\,.$ There are two cases to consider, $j>i$ and $j=i\,.$

\begin{itemize}

\item[$\bullet$] For $j>i\,,$
\begin{eqnarray*}
L^{(j)} h_{(i-1)} & = & \left( L^{(j)} h_{(i)} \right) e^{-\Lambda^{(i)} w^i} \,, \\
& = & - \Lambda^{(j)} h_{(i-1)}\,.
\end{eqnarray*}
This uses Proposition (\ref{udefprop}): the $L^{(j)}$ only contains derivatives with respect to the variables $w^r$ with $r\geq j\,,$ and $j>i$ by assumption.
\medskip

\item[$\bullet$] We first note that if one applies the Lie algebra homomorphism $\rho$ to (\ref{defalpha}) one obtains
\begin{equation}
\Lambda^{(i)} - h_{(i)} \Lambda^{(i)} h_{(i)}^{-1} = \sum_{r>i} \alpha^{(i)}_r ( w_{i+1}\,, \ldots\,,w_{n}) \Lambda^{(r)}\,.
\label{defalpha2}
\end{equation}
This uses the identity
\[
e^A B e^{-A} = B + [A,B] + \frac{1}{2} [A,[A,B]] + \ldots
\]
so applying any Lie algebra homomorphism $\rho([A,B])=[\rho(A),\rho(B)]$ gives
\[
\rho(e^A B e^{-A}) = e^{\rho(A)} \rho(B) e^{-\rho(A)}\,.
\]
Thus {\sl the same} functions $\alpha^{(i)}_r$ appear in both equation (\ref{defalpha}) and (\ref{defalpha2}).

\medskip

With this, for $j=i\,,$
\begin{eqnarray*}
L^{(i)} h_{(i-1)} & = & \left( \frac{\partial~}{\partial w^i} + \sum_{r>i} \alpha^{(i)}_r L^{(r)} \right) . \left( h_{(i)} e^{-\Lambda^{(i)} w^i} \right)\,,\\
& = & - h_{(i)} \Lambda^{(i)} e^{-\Lambda^{(i)} w^i} + \sum_{r>i} \alpha^{(i)}_r \left( L^{(r)} h_{(i)} \right)  e^{-\Lambda^{(i)} w^i}\,,\\
& = & - \left( h_{(i)} \Lambda^{(i)} h_{(i)}^{-1} +\sum_{r>i} \alpha^{(r)}_i \Lambda^{(r)} \right) h_{(i-1)}\,,\\
& = & - \Lambda^{(i)} h_{(i-1)}
\end{eqnarray*}
on using (\ref{defalpha2}).

\end{itemize}

\noindent Hence one may extend the partial solution $h_{(i)}$ to the partial solution $h_{(i-1)}\,,$ moving up the elementary sequence. Repeating this procedure finally yields the function $h_{(0)}$
which satisfies $L^{(i)} h_{(0)} = - \Lambda^{(i)} h_{(0)}$ for all $i>0$, i.e. for $i=1\,,\dots \,, n\,.$
Finally, $\boldsymbol{\xi} = h_{(0)} {\boldsymbol \xi}_0$ is the required solution, where ${\boldsymbol\xi}_0$ is an arbitrary constant vector.

\end{proof}

\noindent Propositions (\ref{udefprop}) and (\ref{GMsolution}) are the central results of this paper: to find the required generalization of the quadratic transformation given in
Theorem \ref{BNtheorem} one has to invert equations (\ref{udef}) to find ${\bf w}={\bf w}({\bf u})$ then then find
$\partial_j v^i = {\boldsymbol\xi}\left({\bf w}(\bf u)\right)$ and integrate, followed by another inversion to find ${\bf u}={\bf u}( {\bf v})\,.$ A simplification occur if the algebra $\mathcal{A}$ has
certain additional properties but before this we give an example (again, taken from \cite{BM1,BM1b}).

\begin{example}\label{nonfaithful} Consider the Novikov algebra given by the multiplication table

\[
\begin{array}{c|ccc}
\circ & e^1 & e^2 & e^3 \\
\hline
e^1 & 0 & e^2 & 2e^3 \\
e^2 & 0 & e^3 & 0\\
e^3 & 0 & 0 & 0 \\
\end{array}
\]

\noindent (this is an example where the representation $\pi_{{}_\Lambda}$ is not a faithful representation of $\mathfrak{g}(\mathcal{A}):$ in this basis $\Lambda^{(3)}={\bf 0}\,)$.
The above formulae give
\begin{eqnarray*}
u^1 & = & e^{2 w^1} \left( 2 w^3 + (w^2)^2 \right) \,,\\
u^2 & = & 2 e^{2 w^1}\,,\\
u^3 & = & e^{2 w^1}
\end{eqnarray*}
and the matrix
\[
\frac{\partial v^i}{\partial u^j} = \left(
\begin{array}{ccc}
1 & 0 & 0 \\
0 & e^{-w^1} & 0 \\
0 & -w^2 e^{-w^1} & e^{-2 w^1}
\end{array}
\right)_{ji}\,.
\]
Integrating yields
\begin{eqnarray*}
v^1 & = & u^1\,,\\
v^2 & = & u^2 \left(u^3\right)^{-\frac{1}{2}}\,,\\
v^3 & = & \log u^3
\end{eqnarray*}
and hence
\begin{eqnarray*}
u^1 & = & v^1\,,\\
u^2 & = & v^2 e^{\frac{1}{2} v^3} \,,\\
u^3 & = & e^{v^3}\,.
\end{eqnarray*}
Under the group of transformation generated by $w^1 \mapsto w^1 + \pi\sqrt{-1}\,, w^2 \mapsto w^2\,, w^3 \mapsto w^3$ (under which $u^i \mapsto u^i e^{2\pi \sqrt{-1}}$, i.e. a monodromy transformation around the discriminant), the $v^i$ transform as:
\[
\left(\begin{array}{c} v^1 \\ v^2 \\ v^3 \end{array}\right) \mapsto
\left(\begin{array}{ccc} 1 & 0 & 0 \\ 0 & -1 & 0 \\ 0 & 0 & 1\end{array}\right)\cdot\left(\begin{array}{c} v^1 \\ v^2 \\ v^3 \end{array}\right)+
\left(\begin{array}{c} 0 \\ 0 \\ 2 \pi\sqrt{-1} \end{array}\right)\,.
\]
\end{example}

\medskip

\noindent A considerable simplification occurs if the algebra has a right-identity.

\medskip

Examples of Novikov algebras with a right-identity are easy to construct. Recall, given a commutative, associative algebra with product $\cdot$ and derivation $\delta\,,$ that
\[
a \circ b = a \cdot b + a \cdot \delta b
\]
defines a Novikov algebra \cite{GD, BM2}. If $e$ is the identity element of the commutative, associative algebra, then $\delta e=0$ and hence $e$ is a right identity for $\circ\,.$ Also, if $\eta$ is a non-degenerate cocycle
for the commutative, associative algebra (i.e. a Frobenius algebra), then it is also a non-degenerate cococyle for the induced Novikov algebra.

\begin{cor} Suppose $\mathcal{A}$ has a right identity. Then
\begin{equation}\label{vresult}
v^i = \eta^{ij} \left(  \prod^{\longrightarrow}_{r\geq 1} e^{+\Lambda^{(r)} w^r} \right)_{j1}\,.
\end{equation}
\end{cor}

\begin{proof} Following \cite{D} (Exercise G.1), $u^1=\frac{1}{2} \eta_{ab} v^a v^b$ (this uses the existence of a right-identity: $e^1 \circ e^1 = e^1$ and hence $g^{11}=2 u^1$). From this
\[
\frac{\partial u^1}{\partial v^j} = \eta_{ja} v^a\,.
\]
But
\[
\left( \prod^{\longleftarrow}_{s=1\,,\ldots\,,n} e^{-\Lambda^{(s)} w^s} \right) \cdot \left( \prod^{\longrightarrow}_{r=1\,,\ldots\,,n} e^{+\Lambda^{(r)} w^r}\right) = \mathbb{I}
\]
and hence
\[
\frac{ \partial u^i}{\partial v^j} = \left( \prod^{\longrightarrow}_{r=1\,,\ldots\,,n} e^{+\Lambda^{(r)} w^r}\right)_{ji}\,.
\]
Combining these gives the results.
\end{proof}

\medskip

Drawing these results together gives:

\begin{theoremA*}
Let $\mathcal{A}$ be a Novikov algebra with a right-identity and satisfying the non-degeneracy conditions in Assumption \ref{assumption}.  The transformation ${\bf u}={\bf u}({\bf v})$ is found by eliminating the ${\bf w}$-variables from the equations
\begin{eqnarray*}
v_i ({\bf w}) & = & \left( \prod^{\longrightarrow} e^{+ \Lambda^{(r)} w^r}\right)_{i1}\,,\\
u^i ({\bf w}) & = & \left( \prod^{\longleftarrow} e^{-\Xi^{(r)} w^r}\right)_{i\bullet}\,.
\end{eqnarray*}
where $v_i=\eta_{ij}v^j\,.$
\end{theoremA*}
Note, since the matrices $\Xi^{(i)}$ and ${\Lambda^{(i)}}$ are upper/lower triangular, the purely algebraic problem of the elimination of the ${\bf w}$-variables is triangular in nature. We now study the process in more detail, and relate the resulting transformations to the monodromy properties of the solutions of the Gauss-Manin equations.

\section{Finite monodromy and polynomial solutions}\label{monodromysection}

By Lie's Theorem the $\Xi^{(i)}$ must be upper triangular, with the matrices becoming increasingly more upper triangular as $i$-increases, until one reaches the end of the sequence. Under some simple assumptions
one may relate this property to the existence of the non-degenerate cocycle.

\begin{prop}\label{assumptions} Suppose that the matrices $\Xi^{(i)}_{rc}$ have non-zero entries on and above the diagonal $c-r=i-1\,,$ and the extra conditions $\Xi^{(i)}_{n+1-i,n}\neq 0\,.$ Then a non-degenerate
cocycle exits with distinguished element $\bullet=n\,.$
\end{prop}

\begin{proof}
From the conditions on $\Xi^{(i)}$, the entries on the bottom $i-1$ rows are zero. Hence $g^{ir}_c=0$ for $r>n+1-i\,,$ or $g^{ij}=0$ for $i+j>n+1\,.$ Similarly, the $i$-th row from the bottom has a single non-zero
entry $\Xi^{(i)}_{n+1-i,n}\,.$ Hence
\[
g^{i,n+1-i} = \left( \Xi^{(i)}_{n+1-i,n}\right) u^n
\]
and so $\Delta({\bf u}) = \left( \prod_{i=1}^n \Xi^{(i)}_{n+1-i,n} \right) (u^n)^n$\,. This is non-zero away from the discriminant $\Sigma=\{ u^n=0 \}\,.$

\medskip

This forces the distinguished element to be $\bullet=n$ (any other would give a degenerate cocycle) and hence the cocycle $\eta^{ij}$ will have non-zero entries on the antidiagonal and zero entries below the anti-diagonal.

\end{proof}

\medskip

For the rest of this section we place various assumptions on the Novikov algebra, namely:

\begin{assumption}\label{assumptions2}
We assume the following:
\begin{itemize}
\item[(i)] the algebra $\mathcal{A}$ has a right-identity;
\item[(ii)] the properties in Proposition \ref{assumptions} hold;
\item[(iii)] the matrices $\Xi^{(1)}$ and $\Lambda^{(1)}$ are purely diagonal.
\end{itemize}

\end{assumption}

\noindent Large numbers of examples - in arbitrary dimensions - may be constructed with these properties.
It is the diagonal matrices $\Xi^{(1)}$ and $\Lambda^{(1)}$ that control the structure of the monodromy group, and conditions on these entries determine
whether or not the monodromy group is finite. Let:
\begin{eqnarray*}
\Lambda^{(1)}_{rc} & = & \lambda_r \delta_{rc}\,,\\
\Xi^{(1)} _{rc}& = & - \mu_r \delta_{rc}
\end{eqnarray*}
(the negative sign is for future notational convenience). These conditions imply that
\begin{eqnarray*}
e^1 \circ e^i & = & \lambda_i e^i\,,\\
e^i \circ e^1 & = & (\mu_i-\lambda_i) e^i
\end{eqnarray*}
with the consistence condition $\mu_1 = 2 \lambda_1\,.$ Thus the first row and first column of the multiplication table for $\circ$ are fully determined. Note:

\medskip

\begin{itemize}
\item[(a)] If $\mathcal{A}$ has a right-identity $e^1$ then $\mu_i=1+\lambda_i$ for all $i\,;$\\
\item[(b)] If $\mathcal{A}$ has an identity $e^1$ (and hence is commutative and associative) then $\lambda_i=1$ and $\mu_i=2$ for all $i\,.$
\end{itemize}

With these conditions, equation (\ref{udef}) may be expanded
\begin{eqnarray*}
{\bf u} & = &  \left( \prod_{r>1}^{\longleftarrow} e^{-\Xi^{(r)} w^r} \right)
\left(
\begin{array}{ccc}
e^{\mu_1 w^1} & \ldots & 0 \\
\vdots & \ddots & \vdots \\
0 & \ldots & e^{\mu_n w^1}
\end{array}
\right)
\left(
\begin{array}{c}
0 \\ \vdots \\ 1
\end{array}
\right)\,,\\
& = & e^{\mu_n w^1} \left( \prod_{r>1}^{\longleftarrow} e^{-\Xi^{(r)} w^r} \right) {\bf u}_0\,.
\end{eqnarray*}
Since the matrices $\Xi^{(r)}$ are, for $r>1$ strictly upper-triangular, the entries above the diagonal are polynomial in the variables $w^2\,,\ldots\,,w^n\,.$ From the triangular structure of the matrices it follow that
\[
u^i({\bf w})  =  e^{\mu_n w^1} f_i(w^{i+1}\,, \ldots \,, w^n)
\]
with the $f_i$ being polynomial functions with $f_n=1\,.$

\medskip

Similarly, equation (\ref{vresult}) may be expanded
\begin{equation}
v^i = \eta^{ik} e^{\lambda_k w^1} \left( \prod^{\longrightarrow}_{s>1} e^{+\Lambda^{(s)} w^s} \right)_{k1}
\label{misc}
\end{equation}
and since the matrices $\Lambda^{(r)}$ are, for $s>1$ strictly upper-triangular, the entries below the diagonal are polynomial in the variables $w^2\,,\ldots\,,w^n\,.$ From the triangular structure of the matrices it follow (where $v_i = \eta_{ij} v^j$) that
\[
v_i({\bf w})  =  e^{\lambda_i w^1} g_i(w^2\,,w^{i+1}\,,\ldots \,, w_i)
\]
with the $g_i$ being polynomial functions with $g_1=1\,.$

From the leading order behaviour one can show (and this uses the $\Xi^{(i)}_{n+1-i,n}\neq 0$ condition) that up to an overall constant
\begin{eqnarray*}
u^n & = & \left(v^n \right)^{\lambda_n+1}\,,\\
u^{n-1} & = & \left(v^n \right)^{\lambda_n-1}\, v^{n-1}\,,\\
\vdots & & \vdots \\
u^i & = & \left(v^n \right)^{\lambda_i}\, v^i + \ldots\,,\\
\vdots & & \vdots \\
u^1 & = & \frac{1}{2} \eta_{ab} v^a v^b\,.
\end{eqnarray*}

\noindent From this structure one can deduce that, under the various assumption, that the $u^i({\bf v})$ are polynomial functions of $v^1\,,\ldots\,, v^{n-1}$. The various powers of $v^n$ depend crucially on the precise values of $\lambda_i\,.$ Note, however, that a necessary condition for polynomial solutions is that the $\lambda_i$ must be integers.

\medskip

To obtain a {\sl finite} monodromy group restrictions on the constants $\lambda_i\,,\mu_n$ have to be imposed. Thus there must exist a smallest integer $N$ such that
\[
N.\frac{\lambda_i}{\mu_n} = k_i \in \mathbb{Z}\,.
\]
With this the monodromy group $\mathcal{W}(\mathcal{A})$ is isomorphic to the cyclic group $\mathbb{Z}_N\,,$ generated by the diagonal matric $A^i_j={\rm diag}(\varepsilon^{k_1}\,, \ldots\,, \varepsilon^{k_n})\,,$
where $\varepsilon=e^{2 \pi\sqrt{-1}/N}$ is an $N$-the root of unity.

\medskip

Bring these results together gives the follow:

\begin{theoremB*}
Suppose that $\lambda_i\in\mathbb{Z}$ for all $i=1\,,\ldots\,,n\,.$ With the conditions on the Novikov algebra contained in Assumption \ref{assumptions2}, the functions $u^i({\bf v})$ are invariant under the monodromy group and
\[
u^i({\bf v}) \in \mathbb{C}^{\mathcal{W}(\mathcal{A})} [ v^1\,,\ldots\,,v^{n-1}, v^n\,,\left(v^n\right)^{-1}]
\]
where the monodromy group is the cyclic group $\mathcal{W}(\mathcal{A})\cong \mathbb{Z}_{\lambda_n+1} [ 1=\lambda_1\,,\ldots\,,\lambda_n]$  which act on the $v_i$ variables by
\[
v_i({\bf w}) \mapsto \varepsilon^{\lambda_i} v_i({\bf w})\,.
\]
where $\varepsilon=e^{\frac{2\pi\sqrt{-1}}{\mu_n}}$ and $\mu_n=\lambda_n+1\,.$
The constants $\lambda_i$ satisfy the monodromy constraint $\lambda_i+\lambda_j=\lambda_n+1$ if $\eta_{ij} \neq 0\,.$

\medskip

In terms of the ${\bf w}$-variables the monodromy group acts as a translation in a single variable
\begin{equation}
\begin{array}{cclc}
w^1 & \mapsto & w^1 + \frac{2 \pi \sqrt{-1}}{\mu_n}\,,&\\
w^i & \mapsto & w^i\,, &\qquad\qquad i=2\,,\ldots\,,n\,.
\end{array}
\end{equation}

\end{theoremB*}

\begin{proof} The invariance of the $u^i$ is immediate. The transformation properties of the $v^i$ follow from the above formulae and the monodromy relation $A^a_b \eta^{bc} A^d_c=\eta^{ad}\,$ gives the result that
$\lambda_i+\lambda_j=\lambda_n+1$ if $\eta_{ij} \neq 0\,.$

\end{proof}

\medskip

The fact that the monodromy is generated by a translation in a {\sl single} variable in the ${\bf w}$-coordinates follows from the fact that the matrices $\Xi^{(i)}$ and $\Lambda^{(i)}$ are, for $i>1\,,$ strictly upper/lower
triangular and hence are nilpotent. Thus their matrix-exponentials results in polynomial functions in the $w^2\,,\ldots\,,w^n$ variables.
It is therefore only $\Xi^{(1)}$ and $\Lambda^{(1)}$ that result in genuine exponentials and these exponentials are invariant under a complex translation: it is these translations that generated the monodromy group.
In the above these matrices $\Xi^{(1)}$ and $\Lambda^{(1)}$ are, by assumption, purely diagonal, but one could easily extend the theory to the case where there are Jordan blocks.

\medskip

One could, in principle, track the dependence of the $u^i({\bf v})$ functions on the $(v^n)^{-1}$-variable, thus deriving conditions under which the functions $u^i({\bf v})$ are polynomial. However as there
is no full classification of Novikov algebras one would obtain conditions that could not be used in any meaningful way. What examples do show is that very subtle cancellations do occur, eliminating this rational dependence
on the $v^n$-variable. Studying such examples suggest the following:

\medskip

\begin{conjecture}\label{Conjecture}
Suppose that $\lambda_i\in\mathbb{N}_{>0}$ for all $i=1\,,\ldots\,,n\,.$ With the conditions on the Novikov algebra contained in Assumption \ref{assumptions2}, the functions $u^i({\bf v})$ are polynomial and are invariant
under the monodromy group:
\[
u^i({\bf v}) \in \mathbb{C}^{\mathcal{W}(\mathcal{A})} [ v^1\,,\ldots\,,v^n]\,.
\]
\end{conjecture}

\noindent This conjecture is true for a wide class of examples, as will be proved in the next section for a class of Novikov algebras which have the monodromy group
$\mathcal{W}({\mathcal{A}}) = \mathbb{Z}_{1+\lambda_n}[1=\lambda_1\,,\ldots\,,\lambda_n]\,.$

\medskip

\section{Cyclic quotient singularities and orbit spaces}\label{quotientsection}

In this section we concentrate on an $n$-dimensional example of a Novikov algebra, and use this to illustrate some of the ideas described above. In particular the relationship between the (finite) monodromy group and
solutions of the Gauss-Manin equations (and their inverses) can be seen very explicitly. The ideas rest heavily on the allied construction of Coxeter group orbits spaces in \cite{D}.

\medskip

Consider the commutative, associative algebra $\mathcal{A} \cong \mathbb{C}[z]/ \langle z^n \rangle\,.$ With the basis $e^i = z^{i-1}\,,i=1\,,\ldots\,, n$ one obtains the commutative, associative algebra
\[
e^i \cdot e^j = e^{i+j-1}
\]
(where we assume $e^j=0$ for $j>n$) with unity $e^1\,.$ The compatible inner product $\langle e^i,e^j \rangle = \delta_{i+j,n+1}$ makes $\{ \mathcal{A}\,,\cdot\,,\langle -,- \rangle \}$ into a Frobenius algebra.
The derivations ${der}(\mathcal{A})$ are easy to compute and one such element is $\partial e^i = (i-1) e^i\,.$ Following \cite{BM2} we may use this to define an $n$-dimensional Novikov algebra
\begin{eqnarray*}
e^i \circ e^j & = & e^i \cdot e^j + e^i \cdot \partial e^j\,,\\
& = & j e^{i+j-1}
\end{eqnarray*}
with associated Lie algebra $[e^i,e^j] = (j-i) e^{i+j-1}\,.$ Thus we have a flat metric
\[
g^{-1} = \sum_{i,j=1}^n (i+j) u^{i+j-1} \frac{\partial~}{\partial u^i} \otimes \frac{\partial~}{\partial u^j}\,.
\]
with the problem of finding the transformation that reduces this to a constant form.

\medskip

Drawing together various remarks from above, this Novikov algebra has a right identity $e^i$ and cocycle $\langle -,- \rangle\,.$ The matrices $\{ \Xi^{(i)}\,, \Lambda^{(i)} \}$ may easily be constructed from the
structure constants of the algebra and these only have non-zero entries on a single diagonal line, so the conditions in Assumption \ref{assumptions2} hold and hence Theorems A and B may be used.
\medskip

Even with these explicit structure constants it would be hard to eliminate the ${\bf w}$-variables from the two sets of the
transformations in Theorem A. However, it turns out that one by-pass this and prove, for
this example, that the $u^i$ are polynomials of degree $(i+1)$ is {\sl all} variables (i.e. not just polynomial in the variables $v^1\,, \ldots\,, v^{n-1}$), but this is at the loss of an explicit formulae for them. The
approached is based on the observation that
\[
L^{(2)} u^i = - (i+2) u^{i+1}
\]
and since we have $u^1 = \frac{1}{2} \eta_{ab} v^a v^b$ one may recursively generate all of the $u^i({\bf v})$ if one can construct the vector field $L^{(2)}$ in the $v$-variables. However, this is just as hard
as eliminating the ${\bf w}$-variables from the
two sets of transformations. However one can show that the coefficients of this vector field - when written in the ${\bf v}$-variables, are quadratic. Thus the application of this vector field results in polynomial
functions $u^i({\bf v})$ of degree $(i+1)$ in the ${\bf v}$-variables.

\medskip

To show this quadratic property we calculate the third derivatives of the coefficients and show these are all zero. This is easily done by calculating the covariant derivatives with respect to the ${\bf u}$-variables, i.e.
\[
\nabla^i \nabla^j \nabla^k L^{(2)}_p
\]
where $\nabla^i = g^{ir} (\,{}^g\nabla_r~)\,.$ Since $L^{(2)}=g^{2r}\frac{\partial~}{\partial u^r}$ one has $L^{(2)}_p=\delta^2_p\,.$ A simple calculation yields
\[
\nabla^k L^{(2)}_p = 2 \delta_{k+1,p}\,, \qquad \nabla^j \nabla^k L^{(2)}_p = 2 \delta_{j+k,p}\,, \qquad \nabla^i \nabla^j \nabla^k L^{(2)}_p=0\,.
\]
Hence in the ${\bf v}$-variables the coefficients $L^{(2)}$ are, at most, quadratic functions. From the relation $[L^{(1)},L^{(2)}] = L^{(2)}$ it follows, since $L^{(1)} = \sum v^i \frac{\partial~}{\partial v^i}$,
that the coefficients are
homogeneous of degree two in the ${\bf v}$-variables and hence are purely quadratic functions. Thus the functions $u^i({\bf v})$ are polynomials of degree $(i+1)$ which are also invariant under the action of the
monodromy group, which in this case is
\begin{equation}
v^i \mapsto \varepsilon^{n+1-i} v^i\,, \qquad i=1\,, \ldots\,, n\,,
\label{action}
\end{equation}
where $\varepsilon$ is the primitive $(n+1)$-th root of unity, $\varepsilon=e^{\frac{2 \pi \sqrt{-1}}{n+1}}\,.$ The monodromy group in this case is the cyclic group $\mathbb{Z}_{n+1}(1\,,2\,,,\ldots\,,n)$ and the
resulting orbit space has a
cyclic quotient singularity. Thus Conjecture \ref{Conjecture} is true for this $n$-dimensional class of examples.

\begin{example}\label{nequalsfour} Consider the case $n=4$. The invariant polynomials are:
\begin{eqnarray*}
u^1 & = & \frac{1}{5} (v^1 v^4 + v^2 v^3)\,,\\
u^2 & = & v^2 \left( v^4\right)^2 + \frac{1}{5} (v^3)^2 v^4\,,\\
u^3 & = & (v^4)^3 v^3\,,\\
u^4 & = & (v^4)^5\,.
\end{eqnarray*}
\end{example}

\medskip

Before explaining some of the geometric structures behind this construction we note that this idea may to applied to general classes of Novikov algebra.

\begin{example}
Consider a Novikov algebra with a two-step solvable Lie algebra $\mathfrak{g}(\mathcal{A})\,.$ From Example \ref{2stepexample} and equation (\ref{misc}) it follows that
\begin{eqnarray*}
L^{(2)} & = & \frac{\partial~}{\partial w^2}\,,\\
& = & \frac{\partial v^i}{\partial w^2} \frac{\partial~}{\partial v^i}\,,\\
& = & \sum_{i,j,k,r} \eta^{ij} \Lambda^{(2)}_{jr} \eta_{rk} \left( v^n\right)^{\lambda_j-\lambda_r} v^k \frac{\partial~}{\partial v^i}\,.
\end{eqnarray*}
Hence, if ${\rm rank}\,\Xi^{(2)} = n-1$ one may generate all of the $u^i({\bf v})$ from $u^1({\bf v})$ by repeated action of $L^{(2)}\,.$ Since $\Lambda^{(2)}$ is strictly lower-triangular it is non-zero only for $j>r\,.$ Thus if the $\lambda_i$ are an increasing sequence of integers it follows that the $u^i({\bf v})$ are polynomial functions.
\end{example}

To understand the geometry we first introduce some notation (following \cite{reid}). Let $k[V]$ be the coordinate ring of $V\cong \mathbb{C}^n$ (the space with coordinates ${\bf v}$). The points in the orbit space (with
coordinates ${\bf u}$)
\[
U = V/ W(\mathcal{A})
\]
correspond to the orbits of the group action. The polynomial functions $u^i({\bf v})$ are $W(\mathcal{A})$-invariant functions, i.e. they belong to the ring $k[V]^{W(\mathcal{A})}\,.$ This ring is well-studied
\[
k[V]^{W(\mathcal{A})} = k[u^1\,, \ldots \,, u^N]/J
\]
so $U$ is a subring of $\mathbb{C}^N$ defined by the ideal $J\,.$ Here we see the difference between this orbit space a Coxeter-group oribit space - by Chevalley's Theorem the later is freely generated by the invariant
polynomial.

\begin{example} Consider a monomial $u=\prod_{i=1}^n v_i^{\alpha_i}$ invariant under the action (\ref{action}). This implies that the $\alpha_i$ have to satisfy the constraint
\begin{equation}
\sum_{i=1}^n (n+1-i) \alpha_i = d (n+1)
\label{constraint}
\end{equation}
for some $d \in \mathbb{N}\,.$ Introducing the basic invariant monomials
\begin{eqnarray*}
u_i & = & v_i v_n^i\,, \quad i = 1 \,, \ldots \,, n-1\,,\\
u_n & = & v_n^{n+1}
\end{eqnarray*}
one may write
\[
u=\left(\prod_{i=1}^{n-1} u_i^{\alpha_i} \right) u_n^{d-\sum_{i=1}^{n-1} \alpha_i}\,.
\]
This result shows that any invariant monomial (and hence invariant polynomial) lies in the ring $\mathbb{C}[u_1\,, \ldots \,, u_n\,, 1/u_n]\,.$

\end{example}

The question of precisely which invariant polynomials make up the function $u^i({\bf v})$ is not answered in this purely algebraic approach. We note, however, that the basic building blocks (so, in example (\ref{nequalsfour}), the monomials
$v^1 v^4\,,v^2 v^3\,,\ldots\,,(v^4)^5$) all have $d=1$ in the constraint (\ref{constraint}) and are, in this sense, the simplest invariant monomials. This observation does not determine the functions $u^i$ - there are more
invariant monomials with $d=1$ than $n\,,$ the dimension of the various spaces.

This description of the spaces $U$ and $V$ has been with reference to the $\bf w$-coordinates that were central to the constructions in Sections \ref{MainSection} and \ref{MainSection2}. What is interesting is that the monodromy group acts,
in the ${\bf w}$-variables, as a simple
affine translation in the {\sl single} variable $w^1$ (this may be traced back to the space of diagonal matrices in $\mathfrak{g}(\mathcal{A})$ is 1-dimensional in this class of examples). Recall, that in these
variables the action is
\[
T\,:\, w^1 \mapsto w^1 + \frac{2\pi \sqrt{-1}}{n+1}\,, \qquad w^i \mapsto w^i\,, \quad i \neq 1\,.
\]
Thus
\[
W \cong \mathbb{C}^{n-1} \times \mathbb{C}/T\,.
\]
Note, as $w^1 \rightarrow \infty\,,$ both ${\bf u}\rightarrow {\bf 0}$ and ${\bf v}\rightarrow {\bf 0}$ so the singular point/discriminant has been taken off to infinity in the ${\bf w}$-picture.

\section{Novikov structures on the cotangent bundle}\label{cotangent}

Under the above generic assumptions, one may generate a pencil of (inverse) flat metrics by applying the transformation $u^\bullet \mapsto u^\bullet + \lambda$ for an arbitrary constant $\lambda\,.$ When applied to the original metric
\begin{equation}
g^{ij}({\bf u}) \mapsto g^{ij}({\bf u}) + \lambda \eta^{ij}
\label{metricpencil}
\end{equation}
and hence, by the Dubrovin-Novikov Theorem, one obtains a biHamiltonian structure. The path from biHamiltonian structures to multiplications on the cotangent bundle has been studied by
many authors \cite{Dpencil,F,M} and here we follow the notation and approach of \cite{DS} which stressed the algebraic structures, and in particular, the Novikov structures, that
appear on the cotangent bundle. With ${}^g\nabla$ and ${}^\eta\nabla$ denoting the Levi-Civita connections of the two flat metrics in (\ref{metricpencil}) one may define a tensorial multiplication
on 1-forms:
\begin{equation}
\alpha \circ \beta = {}^g\nabla_{g^\star\alpha}(\beta) - {}^\eta\nabla_{g^\star \alpha}(\beta)\,, \qquad \alpha\,,\beta \in T^\star{\mathcal M}
\label{diffconnection}
\end{equation}
(this resting on the basic result that the difference of two connections is a tensor). Drawing together various results in \cite{Dpencil,M,DS} gives:

\begin{prop}
The multiplication $\circ: T^\star{\mathcal M} \times T^\star{\mathcal M} \rightarrow T^\star{\mathcal M}$ has the following properties for all $\alpha\,,\beta\,,\gamma\in T^\star{\mathcal{M}}\,:$
\begin{eqnarray*}
(\alpha\circ \beta)\circ \gamma & = & (\alpha \circ \beta) \circ \gamma,\,,\\
\alpha \circ (\beta \circ \gamma) - \beta \circ (\alpha \circ \gamma) & = & (\alpha\circ\beta-\beta\circ\alpha)\circ\gamma\,,\\
g^\star( \alpha \circ \beta, \gamma) & = & g^\star(\alpha, \gamma\circ\beta)\,,\\
\eta^\star( \alpha \circ \beta, \gamma) & = & \eta^\star(\alpha, \gamma\circ\beta)\,.
\end{eqnarray*}
These define a Novikov multiplication on $T^\star{\mathcal{M}}$ compatible with the (two) metrics $g$ and $\eta\,.$
\end{prop}

\noindent In terms of the ${\bf u}$-coordinates (which are flat for the metric $\eta$ defined by (\ref{metricpencil}),
\[
du^i \circ du^j = \Gamma^{ij}_k du^k\,,
\]
but it is important to note that (\ref{diffconnection}) is a coordinate-free definition of the multiplication. We finally note that the endomorphic denoted by $R$ in \cite{Dpencil} is non-invertible for Novikov algebras, and
hence the product $u\cdot v = u \circ R^{-1}(v)\,,$ the essential step in the passage from a Novikov multiplication on the cotangent bundle to a Frobenius structure on the tangent bundle, cannot be made. Finally, note that, as in the Coxeter group orbit space construction, the metric
\[
\eta^{-1} = \eta^{ij} \frac{\partial~}{\partial u^i} \otimes \frac{\partial~}{\partial u^j}
\]
extends across the discriminant where $g^{-1}$ is singular, providing a flat structure on the whole of the orbit space.

\section{Conclusion}\label{conclusion}

The second Hamiltonian structure of the KdV equation also has a multicomponent version - this appearing the in original work of Gelfand and Dorfman \cite{GD}.
This takes the form
\[
{\mathcal{H}}^{ij} =  \left\{ \eta_2^{ij} \frac{d^3~}{dX^3} \right\} + \left\{ (\Gamma^{ij}_r + \Gamma^{ji}_r) u^r \frac{d~}{dX} + \Gamma^{ij}_r u^r_X \right\}
+\left\{ \eta^{ij} \frac{d~}{dX}\right\}\,.
\]
The conditions for this to define a Poisson bracket are again algebraic: one obtains a Novikov algebra with structure constants $\Gamma^{ij}_k$ with cocycle
$\eta^{ij}=\langle e^i,e^j\rangle$ and recall the compatibility condition
\[
\langle a \circ b, c \rangle = \langle a, c \circ b\rangle\,.
\]
The extra requirement comes from the third-order term. On writing $\eta^{ij}_2=\langle e^i, e^j\rangle_2$ this extra condition is
\[
\langle a \circ b , c\rangle_2 = \langle a, b \circ c\rangle_2\,.
\]
If the Novikov algebra is commutative one may take $\eta^{ij}_2=\eta^{ij}$ but in general one has a more restrictive structure. For low-dimensional Novikov algebras - which have been classified
in \cite{BM1,BG} - these cocyles may easily be found \cite{BM1b}. The various terms in brackets in the above expression are all compatible with each other and can be rearranged to form
various biHamiltonian structures and hence integrable hierarchies \cite{SS}.

\medskip

Any local Hamiltonian operator - such as the above third-order operator - may, via the action of the Muira group, be transformed into a constant, or Darboux,
form. This follows from the work of \cite{DZ} and the triviality of certain Poisson cohomology groups \cite{DMS,G}. These cohomology groups describe
the obstructions to the construction of the required Muira transformation, and their triviality shows that no obstructions exist and hence the Muira transformation exists. However, such results do not give these Muira transformations, only their existence. As the results in this paper show, even for first-order operators, where the triviality condition is equivalent to the differential-geometric condition of zero-curvature of a connection, the construction of such Muira transformations - via the solution of the Gauss-Manin equations, is subtle.

\medskip

Thus there remains the problem of the construction of a Muira map for a third-order Hamiltonian operator defined by a Novikov algebra. In the case of a commutative Novikov algebra (i.e. a Frobenius algebra), this was
solved by Balinskii and Novikov \cite{BN}, the transformation is a direct generalization of the original Muira transformation of the KdV equation, namely
\[
u^i = \frac{1}{2} \Gamma^i_{jk} v^j v^k - v^i_X
\]
(here $\Gamma^i_{jk} = \eta_{jr} \Gamma^{ir}_k\,)$. Again, a shorter proof of this may be found in \cite{SZ}.
The results of this paper give the zero-th order term (with respect to the grading defined by the $d_X$-derivatives)  in such a solution, and the triviality of the
cohomology groups imply that there are no obstructions to the calculation of the higher-order terms, just the problem of their explicit calculation.

\medskip

\section*{Acknowledgements}
I would like to thank Prof. Novikov for drawing my attention to the existence of this problem. I would also like to thank Michael Wemyss, Blazej Szablikowski and Dafeng Zuo for various conversations.

\end{document}